\documentclass[12pt,a4paper]{article}

\usepackage{amssymb}
\usepackage{amsmath, amsthm}
\usepackage{arydshln}
\usepackage{epsfig}
\usepackage{setspace}

\usepackage{geometry}

\def\RR{\mathbb{R}}
\def\ZZ{\mathbb{Z}}
\def\QQ{\mathbb{Q}}
\def\FF{\mathbb{F}}
\def\KK{\mathbb{K}}

\numberwithin{equation}{section}

\newcommand{\rank}{\mathop{\rm rank} }

\newcommand{\Det}{\mathop{\rm Det} }
\newcommand{\ncrank}{\mathop{\rm nc\mbox{-}rank} }
\newcommand{\Newton}{\mathop{\rm Newton} }
\newcommand{\ncNewton}{\mathop{\rm nc\mbox{-}Newton} }
\newcommand{\proj}{{\rm proj}}

\newtheorem{Thm}{Theorem}[section]
\newtheorem{Prop}[Thm]{Proposition}
\newtheorem{Lem}[Thm]{Lemma}

\theoremstyle{definition}

\title{A cost-scaling algorithm for computing the degree of determinants}
\author{Hiroshi HIRAI and Motoki IKEDA, \\
Department of Mathematical Informatics, \\
Graduate School of Information Science and Technology,   \\
The University of Tokyo, Tokyo, 113-8656, Japan.\\
\texttt{\normalsize hirai@mist.i.u-tokyo.ac.jp}\\
\texttt{\normalsize motoki\_ikeda@mist.i.u-tokyo.ac.jp}\\
}

\begin{document}
	\maketitle
	\begin{abstract}
    	In this paper, we address computation of the degree $\deg \Det A$ of Dieudonn\'e determinant $\Det A$ of 
    	\[
    	A = \sum_{k=1}^m A_k x_k t^{c_k},
    	\] 
where $A_k$ are $n \times n$ matrices over a field $\KK$, $x_k$ are noncommutative variables, 
	$t$ is a variable commuting with $x_k$, $c_k$ are integers, 
	and the degree is considered for $t$.
	This problem generalizes noncommutative Edmonds' problem and 
	fundamental combinatorial optimization problems  
	including  the weighted linear matroid intersection problem.
	It was shown that $\deg \Det A$ is obtained by
	a discrete convex optimization on a Euclidean building.
	We extend this framework by incorporating a cost scaling technique, 
	and show that $\deg \Det A$  
	 can be computed in time polynomial of $n,m,\log_2 C$, where $C:= \max_k |c_k|$.
   We give a polyhedral interpretation of 	$\deg \Det$, 
   which says that  $\deg \Det A$ is given by linear optimization 
   over an integral polytope with respect to objective vector $c = (c_k)$.
   Based on it, we show that our algorithm becomes a strongly polynomial one.  
	 We also apply our result to an algebraic combinatorial optimization problem 
	 arising from a symbolic matrix having $2 \times 2$-submatrix  structure.
	\end{abstract}
	Keywords: Edmonds' problem, noncommutative rank, Dieudonn\'e determinant, Euclidean building, discrete convex analysis, Newton polytope, partitioned matrix.
\section{Introduction}	
{\em Edmonds' problem}~\cite{Edmonds67} asks to compute the rank of a matrix of the following form: 
\begin{equation}\label{eqn:A}
A = \sum_{k=1}^m A_k x_k,
\end{equation}
where $A_k$ are $n \times n$ matrices over field $\KK$, 
$x_k$ are variables, and $A$ is considered as a matrix over rational function field $\KK (x_1,x_2,\ldots,x_k)$.
This problem is motivated by a linear algebraic formulation of the bipartite matching problem and other combinatorial optimization problems. 
For a bipartite graph $G = ([n] \sqcup [n], E)$, consider $A = \sum_{ij \in E} E_{ij} x_{ij}$, where $E_{ij}$ denotes the $0,1$ matrix having $1$ only for the $(i,j)$-entry. Then $\rank A$ is equal to the maximum size of a matching of $G$.
 Other basic classes of combinatorial optimization problems have such a rank interpretation.
For example, the linear matroid intersection problem corresponds to 
$A$ with rank-1 matrices $A_k$, 
and the linear matroid matching problem corresponds to $A$ with rank-2 skew symmetric matrices $A_k$; see \cite{Lovasz89}.
 
 Symbolical treatment of variables $x_k$ makes the problem difficult, 
 whereas the rank computation after substitution for $x_k$ 
 is easy and it provides the correct value in high probability.
A randomized polynomial time algorithm is obtained by this idea~\cite{Lovasz79}.   
A deterministic polynomial time algorithm for Edmonds' problem 
is not known, and is
one of the important open problems in theoretical computer science.
 
 A recent approach to Edmonds' problem, initiated by Ivanyos et al.~\cite{IQS15a},  is to consider 
 variables $x_k$ to be noncommutative.
 That is,  the matrix $A$ is regarded as a matrix over 
 noncommutative polynomial ring $\KK \langle x_1,x_2,\ldots,x_m \rangle$.
 The rank of $A$ is well-defined via  embedding $\KK \langle x_1,x_2,\ldots,x_m \rangle$ 
 to the {\em free skew field}  $\KK (\langle x_1,x_2,\ldots,x_m \rangle)$.
 The resulting rank is called 
 the {\em noncommutative rank (nc-rank)} of $A$ and is denoted by $\ncrank A$. 
Interestingly, $\ncrank A$ admits a deterministic polynomial time computation:
\begin{Thm}[\cite{GGOW15,IQS15b}]\label{thm:ncrank}
$\ncrank A$ for a matrix $A$ of form {\rm (\ref{eqn:A})} can be computed in time polynomial of $n,m$.
\end{Thm}
The algorithm by Garg et al. \cite{GGOW15} works for $\KK = \QQ$, and
the algorithm by Ivanyos et al. \cite{IQS15b} works for an arbitrary field $\KK$. 
Another polynomial time algorithm for nc-rank is obtained 
by Hamada and Hirai~\cite{HamadaHirai}, while
the bit-length of this algorithm may be unbounded if $\KK = \QQ$.
By the formula of
Fortin and Reutenauer~\cite{FortinReutenauer04}, 
$\ncrank A$ is obtained by  
an optimization problem defined on the family of  vector subspaces in $\mathbb{K}^n$.
The above algorithms deal with this new type of an optimization problem.
It holds $\rank A \leq \ncrank A$, where the inequality can be strict in general.
For some class of matrices including linear matroid intersection, 
$\rank A = \ncrank A$ holds,
and the Fortin-Reutenauer formula provides a combinatorial duality relation.
This  is basically different from the usual derivation by
polyhedral combinatorics and LP-duality.

In the view of combinatorial optimization, 
rank computation corresponds to cardinality maximization.
The degree of determinants is 
an algebraic correspondent of weighted maximization.
Indeed, the maximum-weight of a perfect matching 
of a bipartite graph is equal to 
the degree of the determinant of $\sum_{ij \in E} E_{ij} x_{ij}t^{c_{ij}}$, 
where $t$ is a new variable, $c_{ij}$ are edge-weights, 
and the degree is considered in $t$.
Therefore, the weighed version of Edmonds' problem 
is computation of the degree of the determinant of a matrix $A$ of form~(\ref{eqn:A}), where each $A_k = A_k(t)$ is a polynomial matrix with variable~$t$.

Motivated by this observation and the above-mentioned development, 
Hirai~\cite{HH_degdet} introduced a noncommutative formulation of 
the weighted Edmonds' problem.
In this setting,  the determinant $\det A$ is replaced by 
the {\em Dieudonn\'e determinant $\Det A$}~\cite{Dieudonne} --- 
a determinant concept of a matrix over a skew field.  
For our case, $A$ is viewed as a matrix over the skew field $\FF(t)$
of rational functions with coefficients in 
$\FF = \KK (\langle x_1,x_2,\ldots,x_m \rangle)$.
Then the degree with respect to $t$ is well-defined.
He established a formula of $\deg \Det A$ generalizing the Fortin-Reutenauer formula for $\ncrank A$,  a generic algorithm ({\bf Deg-Det}) to compute 
$\deg \Det A$, and $\deg \det A = \deg \Det A$ relation for weighted linear matroid intersection problem. 
In particular, $\deg \Det$ is obtained in time polynomial of $n$, $m$, 
the maximum degree $d$ of matrix $A$ with respect to $t$, 
and the time complexity of solving the optimization problem for nc-rank.  
Although the required bit-length is unknown for $\KK = \QQ$, 
Oki~\cite{Oki} showed another polynomial time reduction from $\deg \Det$ to 
$\ncrank$ with bounding bit-length. 

In this paper,
we address the $\deg \Det$ computation 
of a special matrix obtained from matrix $A$ (\ref{eqn:A}) 
by assigning ``cost" $c_k$ to each variable $x_k$. 
Namely,  for an integer vector $c=(c_k)_{k \in [m]}$, 
consider 
\begin{equation}\label{eqn:A_weighted}
A[c] := \sum_{k=1}^m A_k x_k t^{c_k}.
\end{equation}
This class of matrices is natural from the view of combinatorial optimization. 
Indeed,  the weighted bipartite matching 
and weighted linear matroid intersection problems
correspond to $\deg \Det$ of such matrices. 
Now exponents $c_k$ of variable $t$ work as weights or costs.   
In this setting, 
the above algorithms~\cite{HH_degdet, Oki} 
are pseudo-polynomial.
Therefore, it is natural to ask for
$\deg \Det$ computation with polynomial dependency in $\log_2 |c_k|$.
The main result of this paper shows that such a computation is indeed possible. 
\begin{Thm}\label{thm:main}
Suppose that arithmetic operations 
over $\KK$ are done in constant time.
Then $\deg \Det A[c]$ for a matrix $A[c]$ of {\rm (\ref{eqn:A_weighted})} can be computed in time polynomial of 
$n,m, \log C$, where $C := \max_k |c_k|$.
\end{Thm}


Our algorithm for Theorem~\ref{thm:main} 
is based on the framework of \cite{HH_degdet}. 
In this framework, $\deg \Det A[c]$ is formulated as a {\em discrete convex optimization 
on the Euclidean building for $GL_n(\KK(t))$}.
The {\bf Deg-Det} algorithm is a simple descent algorithm on the building, where
discrete convexity property ({\em L-convexity})  provides 
a sharp iteration bound of  this algorithm via geometry of the building. 
We incorporate {\em cost scaling} into the {\bf Deg-Det} algorithm,  
which is a standard idea in combinatorial optimization.
To obtain the polynomial time complexity, we need a polynomial sensitivity estimate
for how an optimal solution changes under the perturbation $c_k \rightarrow c_k-1$.
We introduce a new discrete convexity concept, called {\em N-convexity}, that works nicely for such cost perturbation, and show that 
the objective function enjoys this property, from which a desired estimate follows.
This method was devised by \cite{HiraiIkeda} in  
another discrete convex optimization problem on a building-like structure.

We present two improvements of Theorem~\ref{thm:main}.
For this, we give a polyhedral interpretation of $\deg \Det A[c]$, 
which extends a basic fact that $\deg \det A[c]$ is equal 
to the optimal value of the linear optimization 
over the {\em Newton polytope} of $\det A$  
with respect to objective vector $c$.
By utilizing the theory of nc-rank, 
we establish an analogue for $\deg \Det$.
We introduce a noncommutative analogue of 
the Newton polytope, called the {\em  nc-Newton polytope}, and 
show that  $\deg \Det A[c]$ is given by the linear optimization 
over the nc-Newton polytope for $A$. 
The nc-Newton polytope  seems interesting in its own right:
It is a relaxation of the Newton polytope and is an integral polytope. 
As consequences, we obtain:
\begin{itemize}
	\item A strongly polynomial time algorithm 
	by using a preprocessing technique~(Frank and Tardos~\cite{FrankTardos1987}), 
	which rounds $c$ in advance so that $\log C$ is a polynomial of $n,m$, 
	\item A polynomial time algorithm for rational field $\KK = \QQ$
	 by using the modulo $p$ reduction technique (Iwata and Kobayashi~\cite{IwataKobayashi2017}), which reduces $\deg \Det$ computation on 
	 $\QQ$ to that on $GF(p)$ 
	for a polynomial number of smaller primes $p$.
\end{itemize}

As an application, 
we consider an algebraic combinatorial optimization problem
for a symbolic matrix of form
\begin{equation} \label{eqn:2x2}
A = \left(
\begin{array}{ccccc}
A_{11}x_{11} & A_{12} x_{12} &\cdots & A_{1 n} x_{1 n}\\
A_{21}x_{21} & A_{22} x_{22}&\cdots & A_{2 n} x_{2 n} \\
\vdots & \vdots & \ddots & \vdots \\
A_{n1} x_{n1}&A_{n2} x_{n2}&\cdots & A_{n n} x_{n n}
\end{array}\right),
\end{equation}
where $A_{ij}$ is a $2\times 2$ matrix over $\KK$ for $i,j \in [n]$.
We call  such a matrix a {\em $2 \times 2$-partitioned matrix}. 
Rank computation of this matrix is viewed as 
a ``2-dimensional'' generalization of the bipartite matching problem.
The duality theorem by Iwata and Murota~\cite{IwataMurota95} implies $\rank A = \ncrank A$ relation. 
Although $\rank A$ can be computed by the above-mentioned nc-rank algorithms,   
the problem has a more intriguing combinatorial nature.
Hirai and Iwamasa~\cite{HiraiIwamasaIPCO}
showed that $\rank A$ is equal to the maximum size of  
a certain algebraically constrained $2$-matching ({\em $A$-consistent $2$-matching}) 
on a bipartite graph, and  
they developed an augmenting-path type polynomial time algorithm to obtain a maximum $A$-consistent $2$-matching. 
We apply our cost-scaling framework for a $2\times 2$-partitioned matrix $A$ with $x_{ij}$ replaced by $x_{ij}t^{c_{ij}}$, 
and obtain a  polynomial time algorithm to solve
the weighted version of this problem and to compute $\deg \det A (= \deg \Det A)$. 
This result sheds an insight on polyhedral combinatorics, since
it means that linear optimization over the polytope of $A$-consistent $2$-matchings can be solved without knowledge of its LP-formulation.

\paragraph{Related work.}
A matrix $A$ of (\ref{eqn:A})
corresponding to the linear matroid matching problem 
(i.e., each $A_k$ is a rank-2 skew symmetric matrix)
is a representative example in which rank and nc-rank can be different.
Accordingly, $\deg \det$ and $\deg \Det$ can differ 
for a weighted version  $A[c]$ of such $A$.
The computation of  $\deg \det$ of such a matrix is 
precisely the weighted linear matroid matching problem.
Camerini et al.~\cite{CGM1991} utilized this $\deg \det$ formulation and random substitution to obtain a random pseudo-polynomial time algorithm 
solving the weighted linear matroid matching, 
where the running time depends on $C$.
Cheung et al. \cite{CLL2014} speeded up this algorithm,  
and also obtained a randomized FPTAS by using cost scaling.
Recently, Iwata and Kobayashi~\cite{IwataKobayashi2017} 
developed  a polynomial time algorithm solving the weighted linear matroid matching problem, where the running time does not depend on $C$.
The algorithm also uses a similar (essentially equivalent) $\deg \det$ formulation, 
and is rather complicated.
A simplified polynomial time algorithm, possibly using cost scaling, is worthy to be developed,    
in which the results in this paper may help.

\paragraph{Organization.}
The rest of this paper is organized as follows:
In Section~\ref{sec:preliminaries}, 
we give necessary arguments on nc-rank, Dieudonn\'e determinant,  Euclidean building, and discrete convexity.
In Section~\ref{sec:algorithm}, 
we present our algorithm for Theorem~\ref{thm:main}.
In Section~\ref{sec:polyhedral}, we present 
a polyhedral interpretation of $\deg \Det$ and improvements of Theorem~\ref{thm:main}.
In Section~\ref{sec:2x2}, 
we describe the results on $2 \times 2$-partitioned matrices.

\section{Preliminaries}\label{sec:preliminaries}
Let $\RR$, $\QQ$, and $\ZZ$ denote the sets of reals, rationals, and integers, respectively.
Let $e_i \in \ZZ^n$ denote the $i$-th unit vector.
For $s \in [n]:=\{1,2,\ldots,n\}$, let ${\bf 1}_s \in \ZZ^n$ denote  
the 0,1 vector in which the first $s$ components are $1$ and the others are zero, 
i.e., ${\bf 1}_s:= \sum_{i=1}^s e_i$.
We denote ${\bf 1}_n$ by ${\bf 1}$.
For a ring $R$, let $GL_n(R)$ denote the set of $n \times n$ matrices over $R$ 
having inverse $R^{-1}$.
The degree of a polynomial $p(t) = a_k t^{k} + a_{k-1} t^{k-1} + \cdots + a_0$ 
with $a_k \neq 0$ is defined as $k$.
The degree of a rational $p/q$ with polynomials $p,q$ 
is defined as $\deg p - \deg q$.
The degree of the zero polynomial is defined as $-\infty$.

\subsection{Nc-rank and the degree of Dieudonn\'e determinant}
It is known that the rank of matrix $A$ 
over $\KK \langle x_1,\ldots,x_m \rangle$, considered in  $\KK (\langle x_1,\ldots,x_m \rangle)$, 
is equal to the {\em inner rank}---the minimum number 
$r$ for which $A$ is written as $A = BC$ for some
$n \times r$ matrix $B$ and $r \times n$ matrix $C$ over $\KK \langle x_1,\ldots,x_m \rangle$; see \cite{Cohn}.
Fortin and Reutenauer~\cite{FortinReutenauer04} established a formula for 
the inner rank ($=$ the nc-rank) of a matrix $A$ of form (\ref{eqn:A}).
\begin{Thm}[\cite{FortinReutenauer04}]\label{thm:FR}
Let $A$ be a matrix of form~{\rm (\ref{eqn:A})}.
Then $\ncrank A$ is equal to the optimal value of the following problem:
\begin{eqnarray*}
		\mbox{{\rm (R)}} \quad {\rm Min.} && 2n - r - s  \\
		{\rm s.t.} && \mbox{$S A T$ has an $r \times s$ zero submatrix,} \\
		&& S, T \in GL_n(\KK).
	\end{eqnarray*}
\end{Thm}
By $\rank A = \rank SAT \leq 2n - r - s$, nc-rank is an upper bound of rank:
\[
\rank A \leq \ncrank A.
\]
In this paper, we regard the optimal value of (R) as the definition of nc-rank.

\begin{Thm}[\cite{IQS15b}]\label{thm:ncrank_dual}
An optimal solution $S,T$ in {\rm (R)} can be computed in polynomial time. 
\end{Thm}
Notice that the algorithm by Garg et al.~\cite{GGOW15} obtains the optimal value of (R)
but does not  obtain optimal $(S,T)$, and 
that the algorithm by Hamada and Hirai~\cite{HamadaHirai}
obtains optimal $(S,T)$ but
has no guarantee of polynomial bound  of bit-length when $\KK = \QQ$.

Next we consider the degree of the Dieudonn\'e determinant. 
Again we regard the following formula as the definition.
\begin{Thm}[{\cite{HH_degdet}}]
Let $A[c]$ be a matrix of form~{\rm (\ref{eqn:A_weighted})}.
Then $\deg \Det A[c]$ is equal to the optimal value of the following problem:
\begin{eqnarray*}
\mbox{\rm (D)} \quad {\rm Min.} && - \deg \det P - \deg \det Q \\
{\rm s.t.} && \deg (PA_k Q)_{ij} + c_k \leq 0 \quad (i,j \in [n], k \in [m]), \\
&& P,Q \in GL_n (\KK(t)).    
\end{eqnarray*}
\end{Thm}
A pair of matrices $P,Q \in GL_n(\KK(t))$ is said to be
{\em feasible} (resp. {\em optimal}) for $A[c]$ if it is feasible (resp. optimal) to (D) for $A[c]$.
From $0 \geq \deg \det PA[c]Q = \deg \det P + \deg \det Q + \deg \det A[c]$, 
$\deg \Det$ is an upper bound of $\deg \det$:
\[
\deg \det A[c] \leq \deg \Det A[c].
\]

A matrix $M = M(t)$ over $\KK(t)$ is written as a formal power series as
\[
M = M^{(d)} t^{d} + M^{(d-1)} t^{d-1} + \cdots,  
\]
where $M^{(\ell)}$ is a matrix over $\KK$ $(\ell = d,d-1, \ldots)$ and 
$d \geq \max_{ij} \deg M_{ij}$.
For solving (D), the leading term $(PA[c]Q)^{(0)} = \sum_{k} (P A_k t^{c_k}Q)^{(0)} x_k$
plays an important role.
\begin{Lem}[\cite{HH_degdet}]\label{lem:optimality}
Let $(P,Q)$ be a feasible solution for $A[c]$.
\begin{itemize}
\item[{\rm (1)}] $(P,Q)$ is optimal  if and only 
if $\ncrank (PA[c]Q)^{(0)} = n$.
\item[{\rm (2)}]
If $\rank (PA[c]Q)^{(0)} = n$, then $\deg \det A[c] = \deg \Det A[c] = - \deg \det P- \deg \det Q$.
\end{itemize}
\end{Lem}
A direct proof (with regarding (D) as the definition of $\deg \Det$) 
is given in the appendix.

Notice that the optimality condition (1) does not 
imply a good characterization 
(NP $\cap$ co-NP characterization) for $\det \Det A[c]$, 
since the size of $P, Q$ (e.g., the number of terms) 
may depend on $c_k$ pseudo-polynomially.

\begin{Lem}\label{lem:nc_nonsingular}
$\deg \Det A[c] > - \infty$ if and only 
if $\ncrank A  = n$.
\end{Lem}
\begin{proof}
We observe from (D) that $\deg \Det A[c+b{\bf 1}] = nb + \deg \Det A$ and that
$\deg \Det$ is monotone in $c_k$.
In particular, we may assume $c_k \geq 0$.

Suppose that $\ncrank A  < n$.
Then we can choose $S,T \in GL_n(\KK)$ 
such that $SAT$ has 
an $r \times s$ zero submatrix with $r+s > n$ in the upper right corner.
Then, for every $\kappa > 0$, $((t^{\kappa {\bf 1}_r}) S, T(t^{-\kappa {\bf 1}_{n-s}})t^{-C})$ is feasible in (D) with objective value $- \kappa (r+s -n) +nC$, where $C:= \max_k c_k$.
This means that (D) is unbounded.
Suppose that $\ncrank A = n$.
By monotonicity, we have  $\deg \Det A[c]  \geq \deg \Det A$. Now  
$(A)^{(0)} = A$ 
has nc-rank $n$, and $(I,I)$ is optimal by Lemma~\ref{lem:optimality}~(1). 
Then we have 
$\deg \Det A = 0$.
\end{proof}

\subsection{Euclidean building}
Here we explain that the problem (D) is regarded as an optimization over 
the so-called Euclidean building.  See e.g., \cite{Garrett} for Euclidean building.
Let $\KK(t)^-$ denote the subring of $\KK(t)$ 
consisting of elements $p/q$ with $\deg p/q \leq 0$. 
Let $GL_n(\KK(t)^-)$ be the subgroup of 
$GL_n(\KK(t))$ consisting 
of matrices over $\KK(t)^-$  invertible in $\KK(t)^-$.
The degree of the determinant of 
any matrix in $GL_n (\KK(t)^-)$ is zero.
Therefore transformation $(P,Q) \mapsto (LP, QM)$ for $L,M \in GL_n (\KK(t)^-)$ 
keeps the feasibility and the objective value in (D).
Let ${\cal L}$ be the set of right cosets $GL_n (\KK(t)^-) P$ 
of $GL_n (\KK(t)^-)$ in  $GL_n (\KK(t))$, and let ${\cal M}$ be the set of left cosets.

Then (D) is viewed as 
an optimization over ${\cal L} \times  {\cal M}$.
The projection of $P \in GL_n(\KK(t))$ to ${\cal L}$
is denoted by $\langle P \rangle$, which is identified 
with the submodule of $\KK(t)^n$ spanned by the row vectors of $P$ 
with coefficients in $\KK(t)^-$.
In the literature, 
such a module is called a {\em lattice}.
We also denote the projections of $Q$ to ${\cal M}$ by $\langle Q \rangle$ and 
of $(P,Q)$ to ${\cal L} \times {\cal M}$ by $\langle P,Q \rangle$.

The space ${\cal L}$ (or ${\cal M}$) is known as  
the {\em Euclidean building} for $GL_n(\KK(t))$.
We will utilize special subspaces of ${\cal L}$, called {\em apartments}, 
to reduce arguments on ${\cal L}$ to that on $\ZZ^n$. 
For integer vector $\alpha \in \ZZ^n$, 
denote by $(t^{\alpha})$ the diagonal matrix with diagonals $t^{\alpha_1}, t^{\alpha_2},\ldots,t^{\alpha_n}$, that is, 
\[
(t^{\alpha}) =
\left( \begin{array}{cccc}
t^{\alpha_1} &                    &  & \\
                 & t^{\alpha_2}   &  & \\
                 &                     & \ddots & \\
                 &                      &           & t^{\alpha_n} 
\end{array}
\right).
\]
An {\em apartment} of ${\cal L}$ is a subset ${\cal A}$ of ${\cal L}$ represented as
\begin{equation*}
{\cal A} = \{ \langle (t^{\alpha}) P \rangle \mid \alpha \in \ZZ^n \}
\end{equation*}
for some $P \in GL_n(\KK(t))$. The map $\alpha \mapsto \langle (t^{\alpha}) P\rangle $ is an injective map from $\ZZ^n$ to ${\cal L}$.
The following is a representative property of a Euclidean building. 
\begin{Lem}[See \cite{Garrett}]\label{lem:apartment}
For $\langle P \rangle, \langle Q \rangle \in {\cal L}$, there is an apartment containing $\langle P \rangle, \langle Q \rangle$.
\end{Lem}
Therefore ${\cal L}$ is viewed as an amalgamation of integer lattices $\ZZ^n$.
An apartment in  ${\cal M}$ is defined as  a subset of form $\{ \langle Q (t^{\alpha}) \rangle \mid \alpha \in \ZZ^n \}$. 
An apartment in ${\cal L} \times {\cal M}$ is 
the product of 
apartments in ${\cal L}$ and in ${\cal M}$.

Restricting (D) to an apartment 
${\cal A} = \{ \langle (t^{\alpha}) P,  Q(t^{\beta}) \rangle \}_{(\alpha,\beta) \in \ZZ^{2n}}$ 
of ${\cal L} \times {\cal M}$, we obtain a simpler integer program:
\begin{eqnarray*}
{\rm (D_{\cal A}}) \quad
\mbox{Min.} &  &- \sum_{i \in [n]} \alpha_i - \sum_{j \in [n]} \beta_j  + {\rm constant} \\
\mbox{s.t.} & & \alpha_i + \beta_j + c_{ij}^k \leq 0 \quad (k \in [m], i,j  \in [n]), \\
&&  \alpha, \beta \in \ZZ^n,
\end{eqnarray*}
where $c_{ij}^k := \deg (PA_kQ)_{ij} + c_{k}$.
This is nothing but the (discretized) LP-dual of a weighted perfect matching problem.

We need to define a distance between two solutions 
$\langle P, Q \rangle$ and $\langle P', Q' \rangle$  in (D).
Let the {\em $\ell_\infty$-distance} $d_{\infty}(\langle P, Q \rangle, \langle P', Q' \rangle)$ defined as follows:
Choose an apartment ${\cal A}$ containing $\langle P, Q \rangle$ and $\langle P', Q' \rangle$.
Now ${\cal A}$ is regarded as 
$\ZZ^{2n} = \ZZ^n \times \ZZ^n$, and  $\langle P, Q \rangle$ and $\langle P', Q' \rangle$ are regarded as points $x$ and $x'$ in $\ZZ^{2n}$, respectively.
Then define $d_\infty( \langle P, Q \rangle, \langle P', Q' \rangle)$ 
as the $\ell_\infty$-distance $\|x -x'\|_{\infty}$.

The $l_{\infty}$-distance  $d_{\infty}$  is independent of the choice of an apartment, 
and  satisfies the triangle inequality.
This fact is verified 
by applying a {\em canonical retraction} ${\cal L} \times {\cal M} \to {\cal A}$, 
which is distance-nonincreasing; see \cite{HH_degdet}.

\subsection{N-convexity}

The Euclidean building ${\cal L}$ 
admits a partial order in terms of inclusion relation, 
since lattices are viewed as submodules of $\KK(t)^n$. By this ordering, 
${\cal L}$ becomes a lattice in poset theoretic sense; see \cite{HH_degdet,HH_building}.
Then the objective function of (D) is 
a submodular-type discrete convex function on ${\cal L} \times {\cal M}$, 
called an {\em L-convex function}~\cite{HH_degdet}. 
Indeed, its restriction to each apartment ($\simeq \ZZ^{2n}$)
is an L-convex function
in the sense of discrete convex analysis~\cite{DCA}.
This fact played an important role in the iteration analysis of 
the {\bf Deg-Det} algorithm.

Here, for analysis of cost scaling, 
we introduce another discrete convexity concept, 
called {\em N-convexity}.
Since arguments reduce to that on an apartment $(\simeq \ZZ^n)$, 
we first introduce N-convexity on integer lattice $\ZZ^n$.
For $x,y\in \ZZ^n$, let $x \rightarrow y$ be defined by
\begin{equation*}
x  \rightarrow y := x + \sum_{i: y_i > x_i}  e_i - \sum_{i: x_i > y_i}  e_i.
\end{equation*}
Let $x \rightarrow^{i+1} y := (x \rightarrow^{i} y) \rightarrow y$, 
where $x \rightarrow^1 y := x \rightarrow y$. 
Observe that $l_\infty$-distance $\|x - y\|_{\infty}$ 
decreases by one when $x$ moves to $x \rightarrow y$. 
In particular,  $x \rightarrow^d y = y$ if $d = \|x-y\|_{\infty}$.  
The sequence $(x, x\rightarrow^1 y,   x\rightarrow^2 y, \ldots, y)$
is called the {\em normal path} from $x$ to $y$.
Let $y \twoheadrightarrow x$ be defined by
\begin{equation*}
y \twoheadrightarrow x :=  x\rightarrow^{d-1} y  = y + \sum_{i: x_i - y_i = d > 0} e_i -  \sum_{i: x_i - y_i = - d < 0} e_i, 
\end{equation*}
where $d = \|x-y\|_\infty$.

A function $f:\ZZ^n \to \RR \cup \{\infty\}$ is called {\em N-convex} if it satisfies
\begin{eqnarray}
f(x)  + f(y) &\geq & f(x \rightarrow y) + f(y \rightarrow x), \label{eqn:N1} \\
f(x) +  f(y) & \geq &  f(x \twoheadrightarrow y) + f(y \twoheadrightarrow x) \label{eqn:N2}
\end{eqnarray}
for all $x,y \in \ZZ^n$.
\begin{Lem}\label{lem:N-convexity}
\begin{itemize}
\item[{\rm (1)}] $x \mapsto a^{\top} x + b$ is N-convex for $a \in \RR^n, b \in \RR$.
\item[{\rm (2)}]  $x \mapsto \max (x_i + x_j, 0)$ is N-convex for $i,j \in [n]$.
\item[{\rm (3)}] If $f,g$ are N-convex, then $c f + d g$ is N-convex for $c,d \geq 0$. 
\item[{\rm (4)}] Suppose that $\sigma: \ZZ^n \to \ZZ^n$ is a translation 
$x \mapsto x+v$, a transposition of coordinates $(x_1,\ldots,x_i,\ldots,x_j,\ldots,x_n) \mapsto (x_1,\ldots,x_j,\ldots,x_i,\ldots,x_n)$, 
or the sign change of some coordinate $(x_1,\ldots,x_i,\ldots,x_n) \mapsto (x_1,\ldots, - x_i,\ldots,x_n)$.
If $f$ is N-convex, then $f \circ \sigma$ is N-convex.
\end{itemize}
\end{Lem}
\begin{proof}
(1) and (3) are obvious.  
(4) follows from $\sigma (p \rightarrow q) = \sigma (p) \rightarrow \sigma (q)$.
We examine (2). The case of $i=j$ is clear.
We next consider the case of $n=2$ and $(i,j) = (1,2)$.
Let $f(x) :=  \max (x_1 + x_2, 0)$.
Choose distinct $x,y \in \ZZ^2$. 
Let $x' := x \rightarrow y$ (or $x \twoheadrightarrow y$), 
and let  $y' := y \rightarrow x$ (or $y \twoheadrightarrow x$); 
our argument below works for both $\rightarrow$ and $\twoheadrightarrow$.  
We may consider
the case $f(x) < f(x') \in \{ f(x) + 1, f(x) + 2\}$.
We may assume $x_1' = x_1 + 1$.
Then $y_1 \geq x_1' > x_1$.
If $f(x') = f(x) + 2$, then $x_1+x_2 \geq 0$,
$y_2 \geq x_2' >  x_2$, and $y' = y - (1,1)$, implying $f(y') = f(y) - 2$.
Suppose that $f(x') = f(x) + 1$. 
If $x_2' = x_2$, then $y' = y - (1,0)$ and 
$|y_2 - x_2| < y_1 - x_1$, implying 
$y_1+y_2 > x_1+x_2 \geq 0$ and $f(y') = f(y) - 1$.
If $x_2' = x_2+1$, then $x_1 + x_2 = -1$, $y > x$, and $y' = y-(1,1)$.
If $x + (1,1) = y$, 
then $x' = y$ and $y' = x$. 
Otherwise $y_1+y_2 \geq 2$, implying $f(y') = f(y) - 2$
Thus, (\ref{eqn:N1}) and (\ref{eqn:N2}) hold for all cases.

Finally we consider the case $n \geq 3$.
Let $p: \ZZ^n \to \ZZ^2$ be the projection $x \mapsto (x_i,x_j)$. 
Then $f = f \circ p$.
Also it is obvious that $p(x \rightarrow y) = p(x) \rightarrow p(y)$.
Hence $f(x) + f(y) = f(p(x)) + f(p(y)) \geq 
f(p(x) \rightarrow p(y)) + f(p(y) \rightarrow p(x)) = f(p(x \rightarrow y))+  f(p(y \rightarrow x)) =  f(x \rightarrow y)+  f(y \rightarrow x)$.
Also observe that 
$(p(x \twoheadrightarrow y), p(y \twoheadrightarrow x))$ 
is equal to $(p(x),p(y))$ or $(p(x) \twoheadrightarrow p(y), p(y) \twoheadrightarrow p(x))$. From this we have (\ref{eqn:N2}).  
\end{proof}
Observe that the objective function of (D$_{\cal A}$), 
$(\alpha, \beta) \mapsto - \sum_{i=1}^n \alpha_i - \sum_{i=1}^n \beta_i + {\rm const}$ 
if $(\alpha, \beta)$ is feasible, and $\infty$ otherwise,  is N-convex.
A slightly modified version of this fact will be used in the proof of the sensitivity theorem (Section~\ref{subsec:sensitivity}).

N-convexity is definable on ${\cal L} \times {\cal M}$ by taking apartments.
That is, $f:{\cal L} \times {\cal M} \to \RR \cup \{\infty\}$ is called {\em N-convex} if the restriction of $f$ to every apartment is N-convex.
Hence we have the following, though it is not used in this paper explicitly.
\begin{Prop}
The objective function of {\rm (D)} is N-convex on ${\cal L} \times {\cal M}$.
\end{Prop}
In fact, operators $\rightarrow$ and $\twoheadrightarrow$
are independent of the choice of apartments, 
since they can be written by lattice operators on ${\cal L} \times {\cal M}$.

\section{Algorithm}\label{sec:algorithm}
In this section, we develop an algorithm in Theorem~\ref{thm:main}.
In the following, we assume:
\begin{itemize}
\item $\deg \Det A[c] > -\infty$.
\item  Each $c_i$ is a positive integer.
\end{itemize}
The first assumption is verified in advance by nc-rank computation (Lemma~\ref{lem:nc_nonsingular}).
The second one is by  $\deg \Det A[c+b{\bf 1}] = nb + \deg \Det A[c]$.

Also we use the following abbreviation:
\begin{itemize}
\item $A[c]$ is simply written as $A$.
\end{itemize}
\subsection{Deg-Det algorithm}
We here present the {\bf Deg-Det} algorithm~\cite{HH_degdet} for (D),
which is a simplified version of 
Murota's {\em combinatorial relaxation algorithm}~\cite{Murota95_SICOMP} designed for $\deg \det$; see also \cite[Section 7.1]{MurotaMatrix}.
The algorithm uses an algorithm of solving (R) as a subroutine. 

For simplicity, we assume (by multiplying permutation matrices) 
that the position of 
a zero submatrix in (R) is upper right.
\begin{description}
	\item[Algorithm: Deg-Det]
	\item[Input:] $A= \sum_{k =1}^m A_kx_kt^{c_k}$, where $A_k \in \KK^{n \times n}$ and $c_k \geq 1$ for $k \in [m]$, and 
	an initial feasible solution $P,Q$ for $A$.
	\item[Output:] $\deg \Det A$.
	\item[1:]  Solve the problem (R) for $(PAQ)^{(0)}$
	and obtain optimal matrices $S,T$.
	\item[2:] If the optimal value $2n - r - s$ of (R) is equal to $n$, 
	then output $-\deg \det P - \deg \det Q$. Otherwise,
	letting $(P,Q) \leftarrow ((t^{{\bf 1}_{r}})SP, QT(t^{- {\bf 1}_{n-s}}))$, 
	go to step 1.
\end{description}
The mechanism of this algorithm is simply explained:
The matrix $SPAQT$ after step 1 has a negative degree in each entry of  
its upper right $r \times s$ submatrix. 
Multiplying $t$ for the first $r$ rows and $t^{-1}$ 
for the first $n-s$ columns does not produce the entry of degree $>0$.  
 This means that the next solution $(P,Q) := ((t^{{\bf 1}_{r}})SP, QT(t^{- {\bf 1}_{n-s}}))$ is feasible for $A(=A[c])$, 
 and decreases $- \deg \det P - \deg \det Q$ by $r+s - n (> 0)$. 
 Then the algorithm terminates after finite steps, 
 where Lemma~\ref{lem:optimality} (1) guarantees the optimality. 

In the view of Euclidean building,
the algorithm moves the point $\langle P,Q \rangle \in {\cal L} \times {\cal M}$ to an ``adjacent'' point $\langle P',Q' \rangle = \langle (t^{{\bf 1}_{r}})SP, QT(t^{- {\bf 1}_{n-s}}) \rangle$ 
with $d_{\infty}(\langle P,Q \rangle, \langle P',Q' \rangle) = 1$.
Then  the number of the movements ($=$ iterations) is analyzed 
via the geometry of the Euclidean building.
Let ${\rm OPT}(A) \subseteq {\cal L} \times {\cal M}$ denote the set of (the image of) all optimal solutions for $A$.
Then the number of iterations of {\bf Deg-Det} is sharply bounded by
the following distance between from $\langle P,Q \rangle$ to ${\rm OPT}(A)$:
\begin{eqnarray*}
&& \tilde d_{\infty}(\langle P,Q \rangle, {\rm OPT}(A)) :=  \\ 
&& \ \min \{ d_{\infty}(\langle P,Q \rangle, \langle P^*, Q^* \rangle) \mid 
(P^*,Q^*) \in {\rm OPT}(A): 
\langle P \rangle \subseteq  \langle P^* \rangle, \langle Q \rangle \supseteq  \langle Q^* \rangle\},
\end{eqnarray*}
where we regard $\langle P \rangle$ (resp. $\langle Q \rangle$)
as a $\KK(t)^-$-submodule of $\KK(t)^n$ spanned row (resp. column) vectors. 
Observe that $(P,Q) \mapsto (t P, Q t^{-1})$ does not change the feasibility and objective value, and hence an optimal solution $(P^*,Q^*)$ with 
$\langle P \rangle \subseteq \langle P^* \rangle, \langle Q \rangle \supseteq  \langle Q^* \rangle$
always exists. 
%
\begin{Thm}[\cite{HH_degdet}]\label{thm:bound}
The number of executions of step 1 in {\bf Deg-Det} with an initial solution $(P, Q)$ is equal to
$\tilde d_{\infty}(\langle P,Q \rangle, {\rm OPT}(A)) + 1$.
\end{Thm}
This property is a consequence of L-convexity of the objective function of (D). 
Thus {\bf Deg-Det} is a pseudo-polynomial time algorithm.
We will improve {\bf Deg-Det}  by using a cost-scaling technique.

\subsection{Cost-scaling}
In combinatorial optimization, 
cost-scaling is a standard technique to improve 
a pseudo-polynomial time algorithm $\bf A$ to a polynomial one.
Consider the following situation:
Suppose that an optimal solution $x^*$ for costs $\lceil c_k/2 \rceil$ becomes an optimal solution $2x^*$
for costs $2 \lceil c_k/2 \rceil$, and that 
the algorithm $\bf A$ starts from $2x^*$ and obtains
an optimal solution for costs $c_k  \approx 2 \lceil c_k/2 \rceil$ 
within a polynomial number of iterations.
In this case, 
a polynomial time algorithm is obtained by $\log \max_{k} c_k$ calls of $\bf A$.

Motivated by this scenario, we incorporate 
a cost scaling technique with {\bf Deg-Det} as follows:
\begin{description}
	\item[Algorithm: Cost-Scaling]
	\item[Input:] $A = \sum_{k =1}^mA_kx_kt^{c_k}$, where $A_k \in \KK^{n \times n}$ and $c_k \geq 1$ for $k \in [m]$. 
	\item[Output:] $\deg \Det A$.
	\item[0:]  
	Let $C \leftarrow \max_{i \in [m]} c_i$, $N \leftarrow \lceil \log_2 C \rceil$, 
	$\theta \leftarrow 0$, and $(P,Q) \leftarrow (t^{-1} I,I)$. 
	\item[1:] Let $c^{(\theta)}_k \leftarrow \lceil c_i/ 2^{N - \theta} \rceil$ for $k \in [m]$, and let $A^{(\theta)} \leftarrow \sum_{k =1}^{m}A_kx_kt^{c^{(\theta)}_k}$.
	\item[2:] Apply {\bf Deg-Det} for $A^{(\theta)}$  and $(P,Q)$, 
	and obtain an optimal solution $(P^*, Q^*)$ for $A^{(\theta)}$.
	\item[3:] If $\theta=N$, then output $- \deg \det P^* - \deg \det Q^*$.
	Otherwise, 
	 letting $(P, Q) \leftarrow (P^*(t^2), Q^*(t^2))$ and $\theta \leftarrow \theta+1$,  go to step 1. 
\end{description}
For the initial scaling phase $\theta=0$, it holds $c_k^{(0)} = 1$ for all $k$ and
$(P,Q) = (t^{-1}I,I)$ is an optimal solution for $A^{(0)}$ 
(by Lemma~\ref{lem:optimality} and 
the assumption $\ncrank \sum_{k=1}^m A_k x_k= n$).
\begin{Lem}
$(P^*(t^2), Q^*(t^2))$ is an optimal solution for 
$A^{(\theta)}(t^2) = \sum_{k=1}^{m}A_k x_k t^{2 c^{(\theta)}_k}$, and is a feasible solution for $A^{(\theta+1)}$.
\end{Lem}
The former statement follows from the observation that 
the optimality (Lemma~\ref{lem:optimality}~(1)) 
keeps under the change $(P,Q) \leftarrow (P(t^2),Q(t^2))$ and $c_k \leftarrow 2c_k$.
The latter statement follows from the fact that 
$c^{(\theta+1)}_k$ is obtained by decreasing $2 c^{(\theta)}_k$ (at most by $1$). 
The correctness of the algorithm is clear from this lemma.

To apply Theorem~\ref{thm:bound}, 
we need  a polynomial bound of  the distance between the initial solution
$(P^*(t^2), Q^*(t^2))$ of the $\theta$-th scaling phase
and optimal solutions for $A^{(\theta)}$. The main ingredient 
for our algorithm is the  following sensitivity result.
\begin{Prop}\label{prop:sensitivity}
Let $(P,Q)$ be the initial solution in the $\theta$-th scaling phase. 
Then it holds $\tilde d_{\infty}(\langle P,Q \rangle, {\rm OPT} (A^{(\theta)})) \leq n^2 m$.
\end{Prop}
The proof is given in Section~\ref{subsec:sensitivity}, in which 
N-convexity plays a crucial role.
Thus the number of iterations of {\bf Deg-Det} in step 2 is bounded by $O(n^2 m)$, 
and the number of the total iterations is $O(n^2m \log C)$.

\subsection{Truncation of low-degree terms}
Still, the algorithm is not polynomial,
since a naive calculation makes $(P,Q)$ have
a pseudo-polynomial number of terms.
Observe that $(S,T)$ in step 1 of {\bf Deg-Det} depends only on 
the leading term of $PAQ = (PAQ)^{(0)} + (PAQ)^{(-1)}t^{-1} + \cdots$. 
Therefore it is  expected that terms $(PAQ)^{(-\ell)}t^{-\ell}$ with large $\ell>0$
do not affect on the subsequent computation.
Our polynomial time algorithm is obtained 
by truncating such low degree terms.
Note that in the case of the weighted linear matroid intersection, i.e., each $A_k$ is rank-1,
such a care is not needed; see~\cite{FurueHirai, HH_degdet} for details.

First, we present the cost-scaling {\bf Deg-Det} algorithm in the form that 
it updates $A_k$ instead of $P,Q$ as follows:
\begin{description}
	\item[Algorithm: Deg-Det with Cost-Scaling]
	\item[Input:] $A = \sum_{k =1}^{m}A_kx_kt^{c_k}$, where $A_k \in \KK^{n \times n}$ and $c_k \geq 1$ for $k \in [m]$. 
	\item[Output:] $\deg \Det A$.
	\item[0:]  
	Let $C \leftarrow \max_{i \in [m]} c_i$, $N \leftarrow \lceil \log_2 C \rceil$, 
	$\theta \leftarrow 0$, $B_k \leftarrow A_k$ for $k \in [m]$, and  $D^* \leftarrow n$.
    \item[1:] Letting $B \leftarrow \sum_{k=1}^{m} B_k x_k$, 
    solve the problem (R) for $B^{(0)}$
	and obtain an optimal solution $S,T$.
	\item[2:] Suppose that the optimal value $2n - r-s$ of (R) is less than $n$.
	Letting $B_k \leftarrow  (t^{{\bf 1}_{r}}) S B_k T(t^{- {\bf 1}_{n-s}})$ for $k \in [m]$ and
	$D^* \leftarrow D^* + n - r - s$,  go to step 1. 
	\item[3:] Suppose that the optimal value $2n - r-s$ of (R) is equal to $n$.
	If $\theta = N$, then output $D^*$. Otherwise, 
	letting 
	\[
	B_k \leftarrow \left\{   
	\begin{array}{ll}
	B_k(t^2) & {\rm if}\   \lceil c_i/ 2^{N - \theta - 1} \rceil = 2 \lceil c_i/ 2^{N - \theta} \rceil,   
	\\
	t^{-1} B_k(t^2) & {\rm if}\  \lceil c_i/ 2^{N - \theta - 1} \rceil = 2 \lceil c_i/ 2^{N - \theta}\rceil  -1,
	\end{array}\right.
	\]
	$D^* \leftarrow 2D^*$, and $\theta \leftarrow \theta+1$, go to step 1. 
\end{description}
Notice that each $B_k$ is written as the following form:
\[
B_k = B_k^{(0)} + B_k^{(-1)} t^{-1} + B_k^{(-2)} t^{-2} + \cdots,
\]
where $B_k^{(-\ell)}$ is a matrix over $\KK$.
We consider to truncate low-degree terms of $B_k$ after step 1.
For this, we estimate the magnitude of degree 
for which the corresponding term
is irrelevant to the final output.
In the modification $B_k  \leftarrow  (t^{{\bf 1}_{r}}) S B_k T(t^{- {\bf 1}_{n-s}})$ of step $2$, 
the term $B_k^{(-\ell)} t^{-\ell}$ splits into 
three terms of degree $- \ell + 1$, $-\ell$, and $- \ell -1$.
By Proposition~\ref{prop:sensitivity}, 
this modification is done at most $L := mn^2$ time in each scaling phase. 
In the final scaling phase $\theta = N$, 
the results of this phase only depend on terms of $B_k$ with degree at least $- L$.
These terms come from the first $L/2$ terms of $B_k$ 
in the end of the previous scaling phase $\theta=N-1$, 
which come from the first $L/2 + L$ terms of $B_k$ at the beginning of the phase.
They come from the first $(L/2 + L)/2 + L$ terms of the phase $s= N-2$.
A similar consideration shows that the final result is a consequence of the first 
$L (1+1/2 + 1/4 + \cdots + 1/2^{N-\theta}) < 2L$ terms of $B_k$ at the beginning of the $\theta$-th scaling phase. 
Thus we can truncate each term of degree at most $- 2L$: 
Add to {\bf Deg-Det with Cost-Scaling} the following procedure after step $1$.
\begin{description}
\item[Truncation: ] For each $k \in [m]$, 
remove from $B_k$ all terms $B_k^{(- \ell)}t^{- \ell}$ for $\ell \geq 2 n^2m$.
\end{description}
Now we have our main result in an explicit form:
\begin{Thm}
{\bf Deg-Det with Cost-Scaling} computes $\deg \Det A$ in $O( (\gamma(n,m) + n^{2+\omega} m^2) n^2 m \log_2 C)$ time, where  $\gamma(n,m)$ denotes the time complexity of 
solving {\rm (R)} and $\omega$ denotes the exponent of the time complexity of 
matrix multiplication.
\end{Thm}
\begin{proof}
The total number of calls of the oracle solving (R) is  
that of the total iterations $O(n^2 m \log C)$.
By the truncation, the number of terms of $B_k$ is $O(n^2m)$.
Hence the update of all $B_k$ in each iteration is done in $O(n^{2+\omega} m^2)$ time. 
\end{proof}

\subsection{Proof of the sensitivity theorem}\label{subsec:sensitivity}
Let $A = \sum_{k=1}^m A_k x_k t^{c_k}$ and let $A' = A_1x_1 t^{c_1-1} +\sum_{k=2}^m A_k x_k t^{c_k}$.
\begin{Lem}
Let $(P,Q)$ be an optimal solution for $A$.
There is an optimal solution $(P',Q')$ for $A'$ such that $\langle P \rangle \subseteq \langle P' \rangle$, $\langle Q \rangle \supseteq \langle Q' \rangle$, and
$d_{\infty}(\langle P',Q' \rangle, \langle P,Q \rangle) \leq n^2$. 
\end{Lem}
Proposition~\ref{prop:sensitivity} follows from this lemma, 
since $A^{(\theta)}$ is obtained from $A^{(\theta-1)}(t^2)$
by $O(m)$ decrements of $2c_k^{(\theta-1)}$.

Let $(P',Q')$ be an optimal solution for $A'$ such that  $\langle P \rangle \subseteq \langle P' \rangle$, $\langle Q \rangle \supseteq \langle Q' \rangle$, and
$d:= d_{\infty}(\langle P',Q' \rangle, \langle P,Q \rangle)$ is minimum.
Suppose that $d > 0$.
By Lemma~\ref{lem:apartment}, 
choose an apartment ${\cal A}$ of ${\cal L} \times {\cal M}$ 
containing $\langle P,Q \rangle$ and $\langle P',Q' \rangle$.
Regard ${\cal A}$ as $\ZZ^n \times \ZZ^n$. 
Then $\langle P,Q \rangle$ and $\langle P',Q' \rangle$ are regarded as 
points $(\alpha,\beta)$ and $(\alpha',\beta')$ in  $\ZZ^n \times \ZZ^n$, respectively.
The inclusion order $\subseteq$ 
becomes vector ordering $\leq$. 
In particular, $\alpha \leq \alpha'$ and $\beta \geq \beta'$. 
Consider the problem (D$_{\cal A}$) on this apartment.
We incorporate the constraints $x_i + y_i + c_{ij}^k \leq 0$ 
to the objective function as barrier functions.
Let $M > 0$ be a large number. 
Define $h:\ZZ^{n} \times \ZZ^n \to \RR$ by 
\begin{equation*}
h(x,y) := - \sum_{i} x_i - \sum_{i} y_i + M \sum_{i,j, k} \max \{x_i + y_i + c_{ij}^k,0\} \quad ((x,y) \in \ZZ^{n} \times \ZZ^n),
\end{equation*}
where $i,j$ range over $[n]$ and $k$ over $[m]$.
Similarly define $h':\ZZ^{n} \times \ZZ^n \to \RR$ with replacing 
$c_{ij}^{1}$ by $c_{ij}^{1} - 1$ for each $i,j \in [n]$. 

Since $M$ is large, 
$(\alpha, \beta)$ is a minimizer of $h$ and $(\alpha', \beta')$ 
is a minimizer of $h'$.
Note that $(\alpha, \beta)$ is not a minimizer of $h'$. 

Consider the normal path $(z = z^0, z^1,\ldots,z^d = z')$
from $z = (\alpha, \beta)$ to $z' =(\alpha', \beta')$.
Since $z$ and $z'$ satisfy $x_i + y_j + c_{ij}^1  \leq 1$ and
$x_i + y_j + c_{ij}^k \leq 0$ $(k \neq 1)$ for all $i,j \in [n]$,
by N-convexity (Lemma~\ref{lem:N-convexity} (2))
all points $z^{\ell} = (x^{\ell}, y^{\ell})$  in the normal path satisfies these constraints. 
Let $N_\ell$ be the number of the indices 
$(i,j)$ such that $z^{\ell} = (x^{\ell},y^{\ell})$
satisfies  $x_i^{\ell} + y_j^{\ell} + c_{ij}^1 = 1$.
Then 
\begin{equation}\label{eqn:MN_l}
h'(z^{\ell}) = h(z^{\ell}) - M N_{\ell} \quad (\ell = 0,1,2,\ldots,d),
\end{equation}
where  $N_0 = 0$ holds (since $z$ is a feasible solution for $A$).

Next we show the monotonicity of $h,h'$ through the normal path:
\begin{eqnarray}
&& h(z) \leq h(z^1) \leq \cdots  \leq h(z^{d-1}) \leq h(z'), \label{eqn:monotone1}\\
&& h'(z) > h'(z^1) > \cdots > h'(z^{d-1}) > h'(z'). \label{eqn:monotone2}
\end{eqnarray}
Since $h$ is N-convex and $z$ is a minimizer of $h$, 
we have $h(z) + h(z^{\ell}) \geq h(z \twoheadrightarrow  z^\ell ) + h(z^{\ell-1})$
and $h(z) \leq  h(z \twoheadrightarrow z^\ell )$, implying $h(z^{\ell}) \geq h(z^{\ell-1})$.
Similarly, 
since $h'$ is N-convex, 
it holds $h'(z^{\ell}) + h'(z') \geq h'(z^{\ell+1}) + h'(z' \rightarrow z^{\ell})$.
Here $z' \rightarrow z^{\ell} = (\tilde x, \tilde y)$ is closer to $z=(\alpha, \beta)$ 
than $z'$, with $\alpha \leq \tilde x$, $\beta \geq \tilde y$. 
Since  $z'$ is a minimizer of $h'$ nearest to $z$,
 we have $h'(z') < h'(z' \rightarrow z^{\ell})$. Thus $h'(z^{\ell}) > h'(z^{\ell+1})$.

By (\ref{eqn:MN_l}), (\ref{eqn:monotone1}), (\ref{eqn:monotone2}), we have
\begin{equation*}
0 = N_0 < N_1 < \cdots < N_{d-1} < N_d \leq n^2.
\end{equation*}
Thus we have $d \leq n^2$.

\section{Polyhedral interpretation of $\deg \Det$ and its implications}\label{sec:polyhedral}
In this section, we present a polyhedral interpretation of $\deg \Det A[c]$.
Based on it, we show that the previous algorithm becomes strongly polynomial, and
is applicable for the case of $\KK=\QQ$, with avoiding the bit complexity issue.

The starting point is a well-known fact that $\deg \det A[c]$ is
given by the linear optimization 
over the Newton polytope of $\det A$ with objective vector $c$.
Here the {\em Newton polytope} of 
a multivariate polynomial $p(x_1,x_2,\ldots,x_m)= \sum_{u_1,u_2,\ldots,u_m} a_{u_1 u_2 \ldots u_m}x_1^{u_1}x_2^{u_2} \cdots x^{u_m}_m$ 
is defined as the convex hull of all integer vectors $u=(u_1,u_2,\ldots,u_m)$ 
with $a_{u_1 u_2 \ldots u_m} \neq 0$.
Now $\det A$ is a multivariate polynomial of variables $x_1,x_2,\ldots,x_m$. 
Let $\Newton A$ denote the Newton polytope of $\det A$. Then we observe that
\begin{equation}\label{eqn:Newton}
\deg \det A[c] = \max \{ c^{\top} u \mid  u \in \Newton A \}.
\end{equation}
The main theme of this section is to establish an analogous relation 
for $\deg \Det A[c]$.

\subsection{Blow-ups}
We first explain that $\ncrank$ is expressed as the ordinary rank 
of an expanded matrix called a {\em blow-up}.
Let $A = \sum_{k=1}^m A_k x_k$ be a matrix of form (\ref{eqn:A}).
For a positive integer $d$, the $d$-blow up $A^{\{d\}}$ of $A$ is defined by
\begin{equation}
A^{\{d\}} := \sum_{k=1}^m A_k \otimes X_k,
\end{equation}
where $\otimes$ denotes the Kronecker product and
$X_k$ is a $d \times d$ variable matrix 
\[
X_k = 
\left(
\begin{array}{cccc}
x_{k,11} & x_{k,12} & \cdots & x_{k,1d} \\
x_{k,21} & x_{k,22} & \cdots & x_{k,2d} \\
\vdots &     \vdots      &  \ddots         & \vdots        \\             
x_{k,d1}  &     x_{k,d2}     &    \cdots            &  x_{k,dd}        
\end{array}  
\right).
\]
By a natural arrangement of rows and columns, $A^{\{d\}}$ is written as
\begin{equation}\label{eqn:arranged}
A^{\{d\}} = 
\left(
\begin{array}{cccc}
\sum_{k}A_k  x_{k,11} & \sum_{k} A_k   x_{k,12} & \cdots & \sum_{k} A_k  x_{k,1d} \\
\sum_{k} A_k  x_{k,21} & \sum_{k} A_k  x_{k,22} & \cdots & \sum_{k} A_k  x_{k,2d} \\
\vdots &     \vdots      &  \ddots         & \vdots        \\             
\sum_{k} A_k  x_{k,d1}  &    \sum_{k} A_k   x_{k,d2}     &    \cdots            &  \sum_{k} A_k  x_{k,dd}        
\end{array}  
\right). 
\end{equation}
We consider the (ordinary) rank of $A^{\{d\}}$ over 
the rational function field $\KK(\{x_{k,ij}\})$.
It is known \cite{HrubesWigderson2015, Kaliuzhnyi-VerbovetskyiVinnikov2012} that $\ncrank A = n$ if and only if 
$\rank A^{\{d\}} = nd$ for some $d > 0$. 
A linear bound for such $d$ is obtained by:
\begin{Thm}[\cite{DerksenMakam2017}]\label{thm:DerkenMakam}
$\ncrank A = n$ if and only if $\rank A^{\{d\}} =dn$ for $d \geq n-1$.
\end{Thm}
Note that this theorem is stated for infinite field $\KK$. 
In the case of finite field $\KK$,  
each $A_k$ is considered in infinite $\KK(s)$ (say), and
$\ncrank A$ (over $\KK(s)$) does not change 
(by inertia lemma~\cite[Lemma 8.7.3]{Cohn}). 
Then the theorem is applicable.

We show an analogous relation  for $\deg \Det$. 
Let  $c = (c_k)$ is an integer vector. 
The $d$-blow up $A^{\{d\}}[c]$ of  $A[c]$ is defined by
\begin{equation*}
A^{\{d\}}[c] := \sum_{k=1}^m A_k \otimes X_k t^{c_k}.
\end{equation*}

\begin{Lem}\label{lem:degDet=degdet/d} 
\begin{itemize}
\item[{\rm (1)}] 
For $d \geq 1$, it holds
$\displaystyle
\deg \Det A[c] = \frac{1}{d} \deg \Det A^{\{d\}}[c].
$
\item[{\rm (2)}] For $d \geq n-1$, it holds
$\displaystyle
\deg \Det A[c] = \frac{1}{d} \deg \det A^{\{d\}}[c].
$
\end{itemize}
\end{Lem}
\begin{proof}
We may assume that $\ncrank A =n$.
(1). 
Suppose that $A^{\{d\}}[c]$ is arranged as in (\ref{eqn:arranged}).
For  any feasible solution $(P, Q)$ for $A[c]$,  $(P \otimes I_d, Q \otimes I_d)$  is a feasible solution for $A^{\{d\}}[c]$ 
such that $\deg \det P \otimes I_d + \deg \det Q \otimes I_d = d (\deg \det P+ \deg \det Q)$. From this, we have $(\geq)$. 
Conversely, choose any feasible solution $(\tilde P, \tilde Q)$ for $A^{\{d\}}[c]$.
By multiplying $\tilde S \in GL_{nd}(\KK(t)^-)$ to the left of $\tilde P$,
we can assume that the $d \times d$ upper-left submatrix
is nonsingular (over $\KK(t)$) and that the $n(d-1) \times d$ lower-left submatrix 
is the zero matrix.
In addition, replace the $d \times (n-1)d$ upper-right submatrix of $\tilde P$
by the zero matrix.
Then the feasibility and the objective value of (D) do not change. 
This can be seen from the fact that the corresponding submatrices of $A^{\{d\}}[c]$ consist of 
variables $x_{k,ij}$ different from $x_{k,11}, \ldots,x_{k,1d}$.
Repeating this argument (to the $d(n-1) \times d(n-1)$ lower-right block),
we can assume that $\tilde P$ is a diagonal matrix with diagonal blocks $P_1,P_2,\ldots,P_n \in GL_d(\KK(t))$. 
Similarly, $\tilde Q$ is a diagonal matrix with diagonal blocks $Q_1,Q_2,\ldots,Q_n \in GL_d(\KK(t))$. The constraints of (D) are given as 
$\deg (P_\mu A_k Q_\nu)_{ij} + c_k \leq 0$ for $i,j,k, \mu, \nu$.
Clearly the optimum for $A^{\{d\}}[c]$ is attained by $P_1 = P_2 = \cdots = P_n$ and $Q_1=Q_2 = \cdots = Q_n$.
Now $(P_{\mu},Q_{\nu})$ is feasible for $A[c]$.
This concludes $(\leq)$.

(2).  Choose an optimal solution $(P,Q)$ for $A[c]$.
By Lemma~\ref{lem:optimality}~(1) for $A[c]$, 
we have $\ncrank (PA[c]Q)^{(0)} =n$.
Now $d \geq n-1$. 
By Theorem~\ref{thm:DerkenMakam} and $((PA[c]Q)^{(0)})^{\{d\}} =((PA[c]Q)^{\{d\}})^{(0)}$, we have $\rank ((PA[c]Q)^{\{d\}})^{(0)} = nd$.   
As seen above, $(P \otimes I_d, Q \otimes I_d)$ 
is an optimal solution for $A^{\{d\}}[c]$ 
with $((P \otimes I_d) A^{\{d\}}[c] (Q \otimes I_d))^{(0)} = ((PA[c]Q)^{\{d\}})^{(0)}$. 
By Lemma~\ref{lem:optimality}~(2) for $A^{\{d\}}[c]$, 
it holds $\deg \Det A^{\{d\}}[c] = \deg \det A^{\{d\}}[c]$.
With (1), we have the claim.
\end{proof}

\subsection{Nc-Newton polytope}
The determinant of $A^{\{d\}}[c]$ is written as
 \begin{equation}\label{eqn:detA^d[c]}
\det A^{\{d\}}[c] = \sum_{z = (z_{k,ij}) \in \ZZ^{md^2}}
a_z  t^{\sum_{k,i,j} c_k z_{k,ij}} \prod_{k,i,j}x_{k,ij}^{z_{k,ij}},
\end{equation}
where $a_z \in \KK$, $k$ ranges over $[m]$, and $i,j$ range over $[d]$.
For an exponent vector $z= (z_{k,ij})_{k \in[m],i,j \in [d]} \in \ZZ^{md^2}$, 
let $\proj_d (z) \in \QQ^m$ be defined by
\begin{equation}
\proj_d (z)_k = \frac{1}{d} \sum_{i,j \in [d]} z_{k,ij} \quad (k \in [m]).
\end{equation}
Then it holds
\begin{equation}
\deg \det A^{\{d\}}[c] = d \max \{ c^{\top} u \mid u \in \proj_d (\Newton A^{\{d\}}) \}. 
\end{equation}
The {\em nc-Newton polytope} $\ncNewton A$ of $A$ is defined by
\begin{equation}
\ncNewton A := \bigcup_{d=1}^\infty \proj_d (\Newton A^{\{d\}}).
\end{equation}
Analogously to (\ref{eqn:Newton}), 
$\deg \Det A[c]$ is the optimal value of a linear optimization over  $\ncNewton A$.
\begin{Thm} For $c = (c_k) \in \ZZ^m$, 
$\deg \Det A[c]$ is equal to the optimal value of
\begin{eqnarray*}
{\rm LP}[c]: \quad {\rm Max.} && c^{\top} u \\
{\rm s.t.} && u \in \ncNewton A.
\end{eqnarray*}
\end{Thm}
\begin{proof}
Since
$\deg \det A^{\{d\}}[c] \leq \deg \Det A^{\{d\}}[c] = d \deg \Det A[c]$, 
we have $\deg \Det A[c] \geq \max_{d=1,2,\ldots} \max \{ c^{\top} u \mid u \in \proj_d (\Newton A^{\{d\}})\} = 
\max \{ c^{\top} u \mid u \in \ncNewton A \}$.
The equality holds since $\deg \det A^{\{d\}}[c] =(1/d) \deg \Det A[c]$ for $d \geq n-1$ (Lemma~\ref{lem:degDet=degdet/d}~(2)). 
\end{proof}
The proof shows that $\ncNewton A = \proj_d (\Newton A)$ for $d \geq n-1$.
In particular,  $\ncNewton A$ is a rational polytope. 
More strongly, it is an integral polytope.
\begin{Thm}
$\ncNewton A$ is an integral polytope belonging to 
\[
\{ u \in \RR^m \mid u_i \geq 0\ (i=1,2,\ldots,m), \sum_{i=1}^m u_i = n \}.
\]
\end{Thm}
\begin{proof}
For any integral vector $c$, the optimal value of 
linear optimization LP$[c]$ over rational polytope $\ncNewton A$
is given by $\deg \Det A[c]$ that is an integer.
By Edmonds-Giles theorem \cite[Corollary 22.1a]{SchrijverLPIP},  $\ncNewton A$ is an integral polytope.
\end{proof}
It is an interesting direction to study polyhedral combinatorics of $\ncNewton A$.  

\subsection{Strongly polynomial time algorithm}
Now $\deg \Det A[c]$ is interpreted as a linear optimization over 
an integral polytope $\ncNewton A$.
This fact enables us to apply Frank-Tardos method 
to improve a weakly polynomial time algorithm to a strongly polynomial one.

\begin{Thm}[\cite{FrankTardos1987}]
For an integer vector $c = (c_k) \in \ZZ^m$ with $C:= \max_{k} |c_k|$ and a positive integer $N$, 
one can compute an integer vector $\bar c = (\bar c_k) \in \ZZ^m$ 
such that
\begin{itemize}
\item[{\rm (i)}] $\bar C:= \max_{k} |\bar c_k| \leq 2^{4m^3}N^{m(m+2)}$, and
\item[{\rm (ii)}] the signs of $\bar c^{\top} v$ and $c^{\top} v$ are the same for 
all integer vectors $v \in \ZZ^m$ with $\sum_{k=1}^m |v_k| \leq N-1$.
\end{itemize}
In the computation, the number of arithmetic operations 
is bounded by a polynomial of $m$ and 
the required bit-length is bounded 
by a polynomial of $m, \log C, \log N$.
\end{Thm}
We apply this preprocessing for our cost vector $c = (c_k)$ with $N:= mn+1$, and obtain a  modified cost vector $\bar c$.
\begin{Lem}
Any optimal solution of ${\rm LP}[\bar c]$ is also optimal for ${\rm LP}[c]$. 
\end{Lem}
\begin{proof}
Choose any optimal solution $u^*$ for ${\rm LP}[\bar c]$.
For proving the claim, we may assume that $u^*$ is an integer vector.
Consider any other extreme point $u$ in $\ncNewton A$, which is also an integral vector.
In particular, $u,u^*$ are distinct nonnegative integer vectors with 
$u_k^*, u_k \leq n$.
Therefore, $\sum_k |u_k^*-u_k| \leq mn$.
Since $u^*$ is optimal for ${\rm LP}[\bar c]$, 
we have $(\bar c)^{\top}(u^* - u) \geq 0$.
By the property (ii), we have  $c^{\top}(u^* - u) \geq 0$.
This proves the claim.
\end{proof}
For the cost vector $\bar c$, 
the number of the scaling phases of cost scaling {\bf DegDet} algorithm
is $\log_2 \bar C = O(m^3)$.
Therefore, we can compute $\deg \Det A[\bar c]$ 
in strongly polynomial time.
Note that $\deg \Det A[\bar c]$ is different from $\deg \Det A[c]$
but $\deg \Det A[c]$ is given by $c^{\top} u^*$ 
 for any optimal solution $u^*$ for ${\rm LP}[\bar c]$.

We are going to identify $u^*$ by computing $\deg \Det A[c']$
for such approximate vectors~$c'$. 
\begin{Lem}
Let $c' = (c'_k) \in \ZZ^n$.
An optimal solution of ${\rm LP}[(n+1) c' + e_k]$
is precisely an optimal solution $u$ of ${\rm LP}(c')$ having maximum $u_k$. 
\end{Lem}
\begin{proof}
Choose any integral optimal solution  $u^*$ of ${\rm LP}[c']$
and any other integral solution $u$ not optimal to ${\rm LP}[c']$.
Then $(n+1)(c')^{\top} u^* -  (n+1) (c')^{\top} u + (u_k^* - u_k) \geq n+1 + (u_k^* - u_k) > 0$.
This means that $u$ is not optimal to ${\rm LP}[(n+1) c' + e_k]$.
From this we have the claim.
\end{proof}

Let $c^* = (c^*_k) \in \ZZ^m$ be defined by
\begin{equation}
c^* := (n+1)^m \bar c+ (n+1)^{m-1}e_1 + (n+1)^{m-2} e_2 + \cdots + e_m.
\end{equation}
Then ${\rm LP}[c^*]$ has a unique optimal (integral) solution $u^*$ that is lexicographically maximum optimal solution of ${\rm LP}[\bar c]$.
Also $u^*$ is an optimal solution for ${\rm LP}[(n+1)c^* + e_k]$ for all $k$.
Therefore, the coordinates of $u^*$ are determined by
\begin{equation}
 u_k^*  = 
\deg \Det  A[(n+1)c^* + e_k] - (n+1) \deg \Det  A[c^* ]  \quad (k =1,2,\ldots,m). 
\end{equation}
Since $\log_2 \max_k |c_k^*| = O(m^3+ m \log n)$, we have the following:
\begin{Thm}
Suppose that arithmetic operations on $\KK$ are done in constant time. 
Then $\deg \Det A[c]$ and an integral optimal solution of ${\rm LP}[c]$ are computed in strongly polynomial time.
\end{Thm}

\subsection{Polynomial time algorithm for $\KK= \QQ$}
Finally, we consider the case where $\KK = \QQ$.  
In this case, we have to consider 
the bit-complexity for arithmetic operations on $\QQ$.
We avoid this by the method in Iwata and Kobayashi~\cite[Theorem 11.3]{IwataKobayashi2017}.
This method reduces computation over $\QQ$ to that over 
$GF(p)$ for a polynomial number of several (small) primes $p$. 

Suppose that each matrix $A_k$  consists of integer entries 
whose absolute values are at most $D > 0$. 
Then the size of input $A$ is $O(m n^2 \log_2 D)$.
Consider $A_k$ modulo prime $p$,
which is a matrix over $GF(p)$ and is denoted by $(A_k)_{(p)}$.
Consider $A_{(p)}[c] := \sum_{k} (A_k)_{(p)} x_kt^{c_k}$ and  $\deg \Det A_{(p)}[c]$.
By $\deg \Det A_{(p)}[c] = \frac{1}{d}\deg \det A_{(p)}^{\{d\}}[c] = 
\frac{1}{d}\deg ( \det A^{\{d\}}[c]  \mod p)$ for $d \geq n-1$ (Lemma~\ref{lem:degDet=degdet/d}~(2)), 
we have
\begin{equation}
\deg \Det A_{(p)}[c]  \leq \deg \Det A[c]. 
\end{equation}
The equality holds precisely when 
$
a_z \not \equiv 0 \mod p
$ holds in (\ref{eqn:detA^d[c]}) for an exponent vector $z \in \ZZ^{md^{2}}$ 
whose corresponding term in $\det A^{\{d\}}[c]$ 
has the maximum weight relative to $c = (c_k)$.
\begin{Lem}\label{lem:bound}
The absolute value of each coefficient $a_z$ in the expansion $(\ref{eqn:detA^d[c]})$
is bounded by $L := (nd)^{2nd} D^{nd}$. 
\end{Lem}
\begin{proof}
Rename variables $x_{k,ij}$ as $y_k$ $(k=1,2,\ldots,md^2)$.
Accordingly, $A^{\{d\}}$ is rewritten as
\[
A^{\{d\}} = \sum_{k=1}^{md^2} B_k y_k
\]
for $nd \times nd$ matrices $B_k$ over $\QQ$ (in which the absolute value of entries of $B_k$ is bounded by $D$). 
By multilinearity of determinant, we have
\[
\det A^{\{d\}} = \sum_{k_1,k_2,\ldots,k_{nd} \in [md^2]}  \pm \det B[k_1,k_2,\ldots,k_{nd}] y_{k_1}y_{k_2}\cdots  y_{k_{nd}},
\]
where $B[k_1,k_2,\ldots,k_{nd}]$ is the $nd \times nd$ matrix 
with $j$-th row equal to $j$-th of $B_{k_j}$ for $j=1,2,\ldots,nd$. 
For a nonnegative vector $z = (z_k) \in \ZZ^{md^2}$ with $\sum_k z_k = nd$, 
the coefficient $a_z$ of 
$y^{z_1}_{1}y^{z_2}_{2} \cdots y^{z_{md^{2}}}_{md^2}$ in $\det A^{\{d\}}$ is given by
\[
a_z = \sum_{k_1,k_2,\ldots,k_{nd}} \pm \det B[k_1,k_2,\ldots,k_{nd}],
\]
where the sum is taken over all $k_1,k_2,\ldots,k_{nd} \in [md^2]$ such that 
$i \in  [md^2]$ appears $z_i$ times. Hence its absolute value is bounded as
\begin{equation}
|a_z| \leq \frac{(nd) !}{{z_1 !}{z_2 !} \cdots {z_{nd^2}!}}(nd)^{nd}D^{nd} 
\leq (nd)^{2nd} D^{nd}.
\end{equation}
\end{proof}
Take $n-1$ as $d$.
Let $\ell  := \lceil \log_2 L \rceil = O(n^2 \log_2 n + n^2 \log_2 D )$. 
Pick $\ell$ smallest primes $p_1,p_2,\ldots,p_\ell$. 
By the prime number theorem $p_\ell = O(\ell \log \ell)$, 
this can be done (by the sieve of Eratosthenes) in polynomial time.
Now we have
\begin{equation}\label{eqn:Lleq}
L < p_1p_2 \cdots p_\ell.
\end{equation}

Consider an optimal solution $u$ of LP$[c]$.
Then, for $d = n-1$, this $u$ is the projection of some $z \in \ZZ^{md^2}$ 
in which $a_z \neq 0$ in the expansion (\ref{eqn:detA^d[c]}).
With Lemma~\ref{lem:bound} and (\ref{eqn:Lleq}), it holds
\begin{equation}
a_z \not \equiv 0 \mod p_1p_2 \cdots p_\ell.
\end{equation}
Then $a_z \not \equiv 0 \mod p_i$ for some $p_i$.
Therefore, it holds $\deg \Det A_{(p_i)}[c] = \deg \Det A[c]$, and 
we can determine $\deg \Det A[c]$ by 
\[
\deg \Det A[c] = \max_{i=1,2,\ldots,\ell} \deg \Det A_{(p_i)}[c].
\]
This can be done in $\ell$ times computation of $\deg \Det A_{(p_i)}[c]$.
Also $z$ also appears in $\ncNewton (A_{(p_i)})$.
\begin{Thm}
Suppose that $\KK = \QQ$.
Then $\deg \Det A[c]$ and an optimal solution of ${\rm LP}[c]$ can be computed 
in time polynomial of $n, m,  \log C, \log D$.
\end{Thm}

\section{Algebraic combinatorial optimization for 
$2 \times 2$-partitioned matrix}
\label{sec:2x2}
In this section, we consider 
an algebraic combinatorial optimization problem for a  $2 \times 2$-partitioned matrix (\ref{eqn:2x2}).
As an application of the cost-scaling {\bf Deg-Det} algorithm, 
we extend the combinatorial rank computation in~\cite{HiraiIwamasaIPCO} 
 to the deg-det computation.  

We first present the rank formula due to Iwata and Murota~\cite{IwataMurota95} 
in a suitable form for us.
\begin{Thm}[\cite{IwataMurota95}]\label{thm:2x2}
$\rank A$ for a matrix $A$ of form {\rm (\ref{eqn:2x2})} is equal to the optimal value of the following problem:
 \begin{eqnarray*}
	\mbox{\rm (R$_{2 \times 2}$)}\quad {\rm Min.} && 4n - r - s  \\
		{\rm s.t.} && \mbox{$S A T$ has an $r \times s$ zero submatrix}, \\
		&& S, T \in GL_n(\KK),
	\end{eqnarray*}
	where $S,T$ are written as
	\begin{equation}\label{eqn:trans}
S = 
\left(
\begin{array}{cccc}
S_1 & O  & \cdots & O  \\
O & S_2  & \ddots & \vdots \\
\vdots & \ddots & \ddots &   O\\
O & \cdots & O& S_{n}
\end{array}
\right), \quad 
T =
\left(
\begin{array}{cccc}
T_1 & O  & \cdots & O  \\
O & T_2 & \ddots & \vdots \\
\vdots & \ddots & \ddots &   O\\
O & \cdots & O& T_{n}
\end{array}\right)
\end{equation}
for $S_i,T_i \in GL_2(\KK)$ $(i \in [n])$.
\end{Thm}
Namely, (R$_{2\times 2}$) is a sharpening of (R) for $2\times 2$-partitioned matrices, 
where $S,T$ are taken as a form of (\ref{eqn:trans}).
This was obtained earlier than   
the Fortin-Reutenauer formula (Theorem~\ref{thm:FR}).
From the view, 
this theorem implies $\rank A = \ncrank A$ for a $2\times2$-partitioned matrix $A$.
Therefore, by Theorem~\ref{thm:ncrank}, the rank of $A$ can be computed in a polynomial time. 

Hirai and Iwamasa~\cite{HiraiIwamasaIPCO} showed that  
the rank computation of a $2\times2$-partitioned matrix can be formulated as the cardinality maximization problem of certain algebraically constraint $2$-matchings  in a bipartite graph. Based on this formulation and partly inspired 
by the {\em Wong sequence} method~\cite{IQS15a,IQS15b},
they gave a combinatorial augmenting-path type $O(n^4)$-time algorithm 
to obtain a maximum matching and an optimal solution $S,T$ in (R$_{2\times 2}$).

Here, for simplicity of description,  we consider a weaker version 
of this $2$-matching concept.
Let $G_A = ([n] \sqcup [n], E)$ be a bipartite graph defined by $ij \in E \Leftrightarrow A_{ij} \neq O$. 
A multiset $M$ of edges in $E$ is called a {\em $2$-matching} if each node in $G_A$ is incident to at most two edges in $M$.
For a (multi)set $F$ of edges in $E$, 
let $A_F$ denote the matrix obtained from $A$ by replacing 
$A_{ij}$ $(ij \not \in F)$ by the zero matrix. 
Observe that a nonzero monomial $p$ of
a subdeterminant of $A$
gives rise to a $2$-matching $M$ by: 
An edge $ij \in E$ belongs to $M$ with multiplicity $m \in \{1,2\}$ 
if $x_{ij}^{m}$ appears in $p$. Indeed, 
by the $2 \times 2$-partition structure of $A$, 
index $i$ appears at most twice in $p$.
The monomial $p$ also appears in a subdeterminant of $A_M$. 
Motivated by this observation, 
a $2$-matching $M$ is called {\em $A$-consistent} if it satisfies
\[
|M| = \rank (A_M),
\]
where the cardinality $|M|$ is considered as a multiset.
\begin{Prop}[\cite{HiraiIwamasaIPCO}]
$\rank A$ is equal to the maximum cardinality of an $A$-consistent $2$-matching.
\end{Prop}
We see Lemma~\ref{lem:degdetA=c(M)} below for an essence of the proof.
In~\cite{HiraiIwamasaIPCO}, 
a stronger notion of a ($2$-)matching is used, 
and it is shown  that $|M| = \rank (A_M)$ is checked in $O(n^2)$-time 
(by assigning a {\em valid labeling (VL)}). 
An $A$-consistent $2$-matching is called 
{\em maximum} if it has the maximum cardinality 
over all $A$-consistent $2$-matchings.

\begin{Thm}[\cite{HiraiIwamasaIPCO}]\label{thm:HiraiIwamasa}
A maximum $A$-consistent $2$-matching and an optimal solution in {\rm (R$_{2\times 2}$)} can be computed in $O(n^4)$-time.
\end{Thm}
Now we consider a weighted version. 
Consider an integer weight $c = (c_{ij})$, and 
\begin{equation}\label{eqn:2x2_weighted} 
A[c] = \left(
\begin{array}{ccccc}
A_{11}x_{11}t^{c_{11}} & A_{12} x_{12}t^{c_{12}} &\cdots & A_{1 n} x_{1 n}t^{c_{1n}}\\
A_{21}x_{21}t^{c_{21}}  & A_{22} x_{22}t^{c_{22}}&\cdots & A_{2 n} x_{2 n}t^{c_{2n}} \\
\vdots & \vdots & \ddots & \vdots \\
A_{n1} x_{n1}t^{c_{n1}}&A_{n2} x_{n2}t^{c_{n2}}&\cdots & A_{n n} x_{n n}t^{c_{nn}}
\end{array}\right).
\end{equation}
The computation of $\deg \det A[c]$ corresponds to the
maximum-weight $A$-consistent $2$-matching problem.
We suppose that $\rank A = 2n$, and $\deg \det A[c] > -\infty$.
An $A$-consistent matching $M$ (defined for (\ref{eqn:2x2})) is called {\em perfect} 
if $|M| = 2n (= \rank A)$; necessarily such an $M$ is the disjoint union of cycles.
The weight $c(M)$ is defined  by
\begin{equation*}
c(M) = \sum_{ij \in M} c_{ij}.  
\end{equation*}   
Note that $c_{ij}$ contributes to $c(M)$ twice if the multiplicity of $ij$ in $M$ is $2$.
\begin{Lem}\label{lem:degdetA=c(M)}
$\deg \det A[c]$ is equal to the maximum weight of a perfect $A$-consistent $2$-matching.
\end{Lem}
\begin{proof}
Consider the leading term $q \cdot t^{\deg \det A[c]}$ of $\det A[c]$,
where $q$ is a nonzero polynomial of variables $x_{ij}$.
Choose any monomial $p$ in the polynomial $q$. 
As mentioned above, 
the set $M$ of edges $ij$ (with multiplicity $m=1,2$) 
for which $x_{ij}^{m}$ appears in $p$ 
forms a $2$-matching. It is necessarily perfect and $A$-consistent. 
Its weight $c(M)$ is equal to $\deg \det A[c]$.
Thus $\deg \det A[c]$ is at most the maximum weight of a perfect $A$-consistent $2$-matching.

We show the converse. 
Choose a maximum-weight perfect $A$-consistent $2$-matching $M$. 
It suffices to show that $\det A_M[c]$ 
has a nonzero term with degree $c(M)$; such a term also appears in $\det A[c]$.
Now $M$ is a disjoint union of cycles, where a cycle of two (same) edges $ij, ij$
can appear.
We may consider the case where $M$ consists of a single cycle, 
from which the general case follows.
Suppose that $M = \{ij,ij\}$. 
Then $A_{ij}$ must be nonsingular, and 
$\deg \det (A_{ij}x_{ij}t^{c_{ij}}) = 2c_{ij} = c(M)$. 
Suppose that $M$ is a simple cycle of length $2n$.
Then $M$ is the disjoint union of two perfect matchings $M_1, M_2$.
If $A_{ij}$ is nonsingular for all edges $ij$ in the cycle $M$, 
then $M_1$ and $M_2$ are regarded as perfect $A$-consistent $2$-matchings by 
defining the multiplicity of all edges by $2$ uniformly.
By maximality and $c(M) = (c(M_1) + c(M_2))/2$, 
it holds $c(M_1) = c(M_2) = c(M)$.
Replace $M$ by $M_i$. Then $\det A_M[c]$ 
has a single term with degree $c(M)$.    
Suppose that $M_1$ has an edge $ij$ for which $\rank A_{ij} = 1$.	
As in \cite[(2.6)--(2.9)]{HiraiIwamasaIPCO},  
we can take $S_i,T_i \in GL_2(\KK)$ 
such that for each $ij \in M$, 
$A'_{ij} = S_iA_{ij}T_j$ is 
a $2 \times 2$ diagonal matrix with 
$(A'_{ij})_{\kappa \kappa} \neq 0$ if $ij \in M_\kappa$ for $\kappa = 1,2$.  
From $(A'_{ij})_{22} = 0$ for an edge $ij \in M_1$ with $\rank A_{ij}=1$,    
we see that the term of $t^{c(M)}$ 
(obtained by choosing 
the $(\kappa,\kappa)$-element of $A'_{ij}x_{ij} t^{c_{ij}}$ for $ij \in M_{\kappa}, \kappa =1,2$) does not vanish in 
$\det SA[c]T = {\rm const}\cdot \det A[c]$, 
where $S,T$ are block diagonal matrices with diagonals $S_i,T_j$ as in (\ref{eqn:trans}).
\end{proof}	
Corresponding to Theorem~\ref{thm:2x2}, the following holds:
\begin{Lem}
$\deg \det A[c]$ is equal to $\deg \Det A[c]$, which is given by the optimal value of
\begin{eqnarray*}
\mbox{\rm (D$_{2\times 2}$)} \quad {\rm Min.} && - \sum_{i=1}^n \deg \det P_i - \sum_{i=1}^n \deg \det Q_i \\
{\rm s.t.} && \deg (P_i A_{ij} Q_j)_{\kappa \lambda} + c_{ij} \leq 0 \quad (i,j \in [n], \kappa,\lambda = 1,2), \\
&& P_i,Q_j \in GL_2 (\KK(t)) \quad (i,j \in [n]).    
\end{eqnarray*}
In particular, $\Newton A = \ncNewton A$.
\end{Lem}
\begin{proof}
When we apply {\bf Deg-Det} algorithm to $A$ of (\ref{eqn:2x2_weighted}),
$(S,T)$ in the step $1$ is of form of (\ref{eqn:trans}).
Therefore $(PA[c]Q)^{(0)}$ is always of form (\ref{eqn:2x2}), and 
$P$ and $Q$ are always block diagonal matrices with $2 \times 2$ 
block diagonal matrices $P_1,\ldots,P_n$ and $Q_1,\ldots,Q_n$, respectively.
Since $\rank (PA[c]Q)^{(0)} = \ncrank (PA[c]Q)^{(0)}$ (by Theorem~\ref{thm:2x2}),
the output is equal to $\deg \det A[c]$ (Lemma~\ref{lem:optimality}~(2)).
\end{proof}

Now we arrive at the goal of this section.
\begin{Thm} 
Suppose that arithmetic operations on $\KK$ are done in constant time.
A maximum-weight  perfect $A$-consistent $2$-matching 
(and $\deg \det A[c]$) can be computed in 
$O(n^6 \log C)$-time, where $C := \max_{i,j \in [n]} |c_{ij}|$. 
\end{Thm}
\begin{proof}
Apply {\bf Deg-Det with Cost-Scaling} to the matrix $A$.
Since $A_{ij}$ is $2\times 2$, 
$N_d$ in the proof of 
the sensitivity theorem (Section~\ref{subsec:sensitivity}) can be taken to be $4$ (constant), 
whereas $m$ is $n^2$.
Therefore, in each scaling phase, 
the number of iterations is bounded by $n^2$.
Then the degree bound for truncation is chosen as $2n^2$.
The time complexity for matrix update is $O(n^2 \times n^2)$; 
this is done by matrix multiplication of $2 \times 2$ matrices. 
By Theorem~\ref{thm:HiraiIwamasa}, $\gamma(n,m) = O(n^4)$.
The total time complexity is $O(n^6 \log C)$.

Next we find a maximum-weight perfect $A$-consistent 2-matching 
from the final $B^{(0)}$ for $B= B^{(0)} + B^{(-1)}t^{-1} + \cdots$.
Consider a maximum $B^{(0)}$-consistent 2-matching $M$ for $2 \times 2$-partitioned matrix 
$B^{(0)}$ (of form (\ref{eqn:2x2})). 
Necessarily $M$ is perfect (since $B^{(0)}$ is nonsingular).
We show that $M$ contains a maximum-weighted $A$-consistent $2$-matching.
Indeed, $B^{(0)}$ is equal to $(PAQ)^0$ for 
$P,Q \in GL_n(\KK(t))$, where $P$ and $Q$ are block diagonal matrices 
with $2 \times 2$ block diagonals $P_1,P_2,\ldots,P_n$ 
and $Q_1,Q_2,\ldots,Q_n$. 
Notice that $P_i,Q_j$ are an optimal solution of (D$_{2\times 2}$).
Observe $B^{(0)}_M = (P A_M Q)^0$. 
From this,we have $\deg \det PAQ \geq \deg \det P A_M Q 
=  \deg \det A_M + \sum_{i} \deg \det P_i + \sum_{i} \deg \det Q_i = \deg \det B^{(0)}_M = 0$.
This means that $\deg \det A_M$ is equal to $\deg \det A$,  
which is the maximum-weight of a perfect $A$-consistent $2$-matching (Lemma~\ref{lem:degdetA=c(M)}).
Therefore, $M$ must contain a maximum-weight perfect $A$-consistent $2$-matching.
It is easily obtained as follows.
Consider a simple cycle $C = C_1 \cup C_2$ of $M$, 
where $C_1$ and $C_2$ are disjoint matchings in $C$. 
For $\kappa \in \{1,2\}$,  
if $C_\kappa$ 
consists of edges $ij$ with $\rank A_{ij} = 2$ and 
$c(C_\kappa) \geq c(C)$, then replace $C$ by $C_\kappa$ in $M$.
Apply the same procedure to each cycle.
The resulting $M$ satisfies $c(M) = \deg \det A_M$, as desired.
\end{proof}
According to the machinery in the previous section,  
one can make this algorithm strongly polynomial.
Also, for the case of $\KK = \QQ$,  a maximum-weight  perfect $A$-consistent $2$-matching can be obtained in polynomial time.

From the view of polyhedral combinatorics, 
it is a natural question to ask for
the LP-formulation describing 
the polytope of $A$-consistent $2$-matchings (or more generally, nc-Newton polytopes). 
One possible approach to this question is
to clarify the relationship between 
the LP-formulation and (R$_{2\times 2}$).

\section*{Acknowledgments}
The authors thanks Kazuo Murota for comments.
The first author was supported by JSPS KAKENHI Grant Numbers JP17K00029
and JST PRESTO Grant Number JPMJPR192A, Japan.

\section*{Appendix: Proof of Lemma~\ref{lem:optimality}}
(1). We have seen the only-if part
in the explanation of {\bf Deg-Det}.
So we show the if part. 
We first extend $\deg \Det B$ 
for matrix $B = \sum_{k=1}^m B_k (t) x_k$, 
where $B_k(t)$ are matrices over $\KK(t)$.
This is naturally defined by (D) 
in replacing the constraint by 
$\deg (PB_kQ)_{ij} \leq 0$.
In this setting, it obviously holds that 
$\deg \Det PBQ = \deg \det P + \deg \det Q + \deg \Det B$.
Therefore  it suffices to show $\deg \Det B = 0$ 
if $\deg B_{ij} \leq 0$ for all $i,j$ and
$\ncrank B^{(0)} = n$.
 
Let $(P,Q)$ be any feasible solution for $B$.
Recall  the Smith-McMillan form that $P,Q$ are written 
as $P  =S' (t^{\alpha}) S$, $Q = T (t^{-\beta}) T'$ for $S,S',T,T' 
\in GL_n(\KK(t)^-)$, $\alpha,\beta \in \ZZ^n$.
Since the multiplication of $S',T'$ 
does not change the feasibility and the objective value,
we can assume that $P,Q$ are 
form of $P = (t^{\alpha}) S$, $Q = T (t^{-\beta})$.
We can assume further that 
$\alpha_1 \geq \alpha_2 \geq \cdots \geq \alpha_n \geq 0$
and $\beta_1 \geq \beta_2 \geq \cdots \geq \beta_n \geq 0$.
Note that $S^{(0)}$, $T^{(0)}$ are nonsingular matrices over $\KK$.
From $\deg (PBQ)_{ij} \leq 0$, 
it must hold that 
$\alpha_i > \beta_j$  implies  $(S^{(0)} B^{(0)}T^{(0)})_{ij} = 0$.
Let $0 =: \gamma_0 \leq \gamma_1 < \gamma_2 < \cdots < \gamma_{\ell}$ so that $\{ \gamma_1,\gamma_2, \ldots, \gamma_{\ell} \} = 
\{\alpha_i\}_{i=1}^n \cup \{\beta_j\}_{j=1}^n$.
For each $p=1,2,\ldots, \ell$, 
define the indices 
$r_p:= \max \{i \mid \alpha_i \geq \gamma_p\}$
and $u_p = \min \{j \mid \gamma_{p-1} \geq \beta_j \}$.
Then $S^{(0)} B^{(0)}T^{(0)}$ must have 
an $r_p \times (n-u_p +1)$ zero submatrix in 
its upper right corner.
Since $\ncrank B^{(0)} = n$, it holds
\begin{equation*}
- r_p + u_p - 1 \geq 0.
\end{equation*}
Also, $\alpha, \beta$ are written as
\begin{equation*}
\alpha = \sum_{p=1}^\ell (\gamma_p - \gamma_{p-1}) {\bf 1}_{r_p}, \quad
\beta =  \sum_{p=1}^\ell (\gamma_p - \gamma_{p-1}) {\bf 1}_{u_p-1}.
\end{equation*}
Now $- \deg \det P - \deg \det Q$ is equal to
\[
-\sum_{i=1}^n \alpha_i + \sum_{j=1}^n \beta_j = 
\sum_{p =1}^{\ell}  (\gamma_p - \gamma_{p-1}) (- r_p + u_p - 1) \geq 0.
\] 
This means that every feasible solution has the objective value at least $0$, 
and $(I,I)$ is an optimal solution for $B$, implying $\deg \Det B = 0$. 

(2). It holds $\deg \det P A[c] Q = \deg \det P + \deg \det Q + \deg \det A$.
If $\rank (PA[c]Q)^{(0)} =n$ and  $\deg (P A[c] Q)_{ij} \leq 0$ for $i,j$, 
it holds $\deg \det P A[c] Q = 0$.
In this case, it also holds $\ncrank (PA[c]Q)^{(0)} =n$, and hence 
$\deg \Det P A[c] Q = 0$, implying  $\deg \Det A[c] = - \deg \det P - \deg \det Q$.  	
\end{document}